\documentclass[10pt, conference, compsocconf]{IEEEtran}
\IEEEoverridecommandlockouts
\usepackage{latexsym,epsfig,color}
\usepackage{amsmath,amssymb}

\newtheorem{theorem}{Theorem}
\newtheorem{definition}{Definition}

\def\A{{\mathcal A}}
\def\C{{\mathcal C}}

\def\F{{\mathcal F}}

\def\K{{\mathcal K}}

\def\R{{\mathcal R}}
\def\PP{{\mathcal P}}

\def\V{{\mathcal V}}
\def\XA{{\stackrel{\circ}{\mathcal A}}}
\def\XP{{\stackrel{\circ}{\mathcal P}}}
\def\XS{{\stackrel{\circ}{\mathcal S}}}
\newcommand{\eps}{\varepsilon}

\begin{document}
%
\title{On 2-Site Voronoi Diagrams under Geometric Distance Functions
\thanks{\copyright 2011 IEEE. Personal use of this material is permitted. Permission from IEEE must be obtained for all other uses, in any current or future media, including reprinting/republishing this material for advertising or promotional purposes, creating new collective works, for resale or redistribution to servers or lists, or reuse of any copyrighted component of this work in other works.}
}


\author{
   \IEEEauthorblockN{
      Gill Barequet\IEEEauthorrefmark{1},
      Matthew T. Dickerson\IEEEauthorrefmark{2},
      David Eppstein\IEEEauthorrefmark{3},
      David Hodorkovsky\IEEEauthorrefmark{4},
      Kira Vyatkina\IEEEauthorrefmark{5}\IEEEauthorrefmark{6}
   }
   \IEEEauthorblockA{
      \IEEEauthorrefmark{1}
      Dept.\ of Computer Science,
      The Technion---Israel Institute of Technology,
      Haifa~32000, Israel \\
      E-mail: {\tt barequet@cs.technion.ac.il}
   }
   \IEEEauthorblockA{
      \IEEEauthorrefmark{2}
      Dept.\ of Mathematics and Computer Science,
      Middlebury College, Middlebury, VT~05753 \\
      E-mail: {\tt dickerso@middlebury.edu}
   }
   \IEEEauthorblockA{
      \IEEEauthorrefmark{3}
      Dept.\ of Information and Computer Science,
      University of California, Irvine, CA~92717 \\
      E-mail: {\tt eppstein@ics.uci.edu}
   }
   \IEEEauthorblockA{
      \IEEEauthorrefmark{4}
      Dept.\ of Applied Mathematics,
      The Technion---Israel Institute of Technology \\
      Haifa~32000, Israel
   }
   \IEEEauthorblockA{
      \IEEEauthorrefmark{5}
      Dept.\ of Mathematics and Mechanics,
      Saint Petersburg State University, \\
      28 Universitetsky pr., Stary Peterhof,
      St.\ Petersburg~198504, Russia \\
      E-mail: \texttt{kira@math.spbu.ru}
   }
   \IEEEauthorblockA{
      \IEEEauthorrefmark{6}
      Dept.\ of Natural Sciences,
      Saint Petersburg State University of Information \\
         Technologies, Mechanics and Optics,
      49 Kronverkskiy pr., St.\ Petersburg~197101, Russia
   }
}

\IEEEspecialpapernotice{(IEEE CS, Proceedings of the 8th ISVD, Qingdao, China, June 28-30, 2010; ACCEPTED)}

\maketitle


\begin{abstract}
   We revisit a new type of a Voronoi diagram, in which
   distance is measured from a point to a \emph{pair} of points.
   We consider a few more such distance functions, based on geometric primitives, and analyze the structure and complexity of the nearest- and furthest-neighbor Voronoi
   diagrams of a point set with respect to these distance functions.
\end{abstract}


%
\IEEEpeerreviewmaketitle


\section{Introduction}

The Voronoi diagram is one of the most fundamental concepts in computational
geometry, which has plenty of applications in science and industry.  Much
information in this respect can be found in~\cite{Au91} and~\cite{OBS00};
for important recent achievements, see~\cite{G08}.

The basic definition of the Voronoi diagram applies to a set~$S$ of~$n$
points (also called \emph{sites}) in the plane: its \emph{nearest-neighbor}
Voronoi diagram~$V(S)$ is a partition of the plane into~$n$ regions, each
corresponding to a distinct site $s\in S$, and consisting of all the points
being closer to~$s$ than to any other site from~$S$.  Similarly, the
\emph{furthest-neighbor} Voronoi diagram of~$S$ is obtained by assigning
each point in the plane to the region of the most remote site.  These
notions can be generalized to higher-dimensional spaces, different types of
sites, and in other ways.

One of the recent generalizations of this concept is a family of so-called
\emph{2-site Voronoi diagrams}~\cite{BDD02}, which are based on distance
functions that define a distance from a point in the plane to a \emph{pair}
of sites from a given set~$S$.  Consequently, each Voronoi region corresponds
to an (unordered) pair of sites from~$S$.
The original motivation for the study~\cite{BDD02} was the famous
Heilbronn's triangle problem~\cite{Ro51}.  Other motivations are mentioned
therein.

For~$S$ being a set of points, Voronoi diagrams under a number of 2-site
distance functions have been investigated, which include arithmetic
combinations of point-to-point distances~\cite{BDD02,VB10} and certain
geometric distance functions~\cite{BDD02,DE09,HB09}.  In this work, we
develop further the latter direction.

Let $S\subset \mathbb{R}^2$, and consider $p,q\in S$ and a point~$v$ in the
plane.  We shall focus our attention on a few circle-based distance
functions:
\begin{itemize}
\item \emph{radius of circumscribing circle}:
      $\C(v,(p,q)) = \mathrm{Rad}(\circ(v,p,q))$, where $\circ(v,p,q)$ is
      the circle defined by $v,p,q$ and~$\mathrm{Rad}(c)$ is the radius of
      the circle~$c$;
\item \emph{radius of containing circle}:
      $\K(v,(p,q)) = \mathrm{Rad}(C(v,p,q))$, where $C(v,p,q)$ is
      the minimum circle containing $v,p,q$;\footnote{
         Obviously, $\circ(v,p,q) \neq C(v,p,q)$ if any of the three
         points is properly contained in the circle whose diameter is
         defined by the two other points.
      }
\item \emph{view angle}:
      $\V(v,(p,q)) = \measuredangle{pvq}$, or, equivalently, half of the
      angular measure of the arc of $\circ(v,p,q)$ that the angle
      $\measuredangle{pvq}$ subtends;
\item \emph{radius of inscribed circle}:
      $\R(v,(p,q))$ is the radius of the circle inscribed in
      $\triangle(v,p,q)$;
\item \emph{center-of-circumscribing-circle-based functions}:
      let $o_{vpq}$ denote the center of the circle $\circ(v,p,q)$; then
      $\XS(v,(p,q))$, $\XA(v,(p,q))$, and $\XP(v,(p,q))$ are the
      distance from $o_{vpq}$ to the segment $pq$, the area of
      $\triangle o_{vpq}pq$, and the perimeter of $\triangle o_{vpq}pq$,
      respectively;
\end{itemize}
and on a parameterized perimeter distance function:
\begin{itemize}
\item \emph{parameterized perimeter}:
       $\PP_c(v,(p,q))=|vp|+|vq|+c \cdot |pq|$, where $c \geq -1$.
\end{itemize}
The first and third circle-based distance functions were first mentioned
in~\cite{H05}.  The last function generalizes the perimeter distance function
$\PP(v,(p,q)) = \mathrm{Per}(\triangle(v,p,q)$ introduced in~\cite{BDD02},
and later addressed in~\cite{DE09,HB09}.

Since two points define a segment, any 2-point site distance function
$d(v,(p,q))$ provides a distance between the point $v$ and the segment $pq$,
and vice versa.  Consequently, geometric structures akin to 2-site Voronoi
diagrams can arise as Voronoi diagrams of segments.  This alternative approach
was independently undertaken by Asano et al., and the ``view angle'' and
``radius of circumscribing circle'' distance functions reappeared in their
works~\cite{AKTT06,AKTT07} on Voronoi diagrams for segments soon after they
had been proposed by Hodorkovsky~\cite{H05} in the context of 2-site Voronoi
diagrams.  However, as Asano's et al.\ research was originally motivated by
mesh generation and improvement tasks, they were mostly interested in sets of
segments representing edges of a simple polygon, and thus, non-intersecting
(except, possibly, at the endpoints), what significantly alters the essence
of the problem.

In this paper, we analyze the structure and complexity of 2-site Voronoi
diagrams under the distance functions listed above.  Our obtained results are
mostly of theoretical interest.  The method used to derive an upper bound on
the complexity of the nearest-neighbor 2-site Voronoi diagram under the ``parameterized
perimeter'' distance function is first developed for the case of $c=1$,
yielding a much simpler proof for the ``perimeter'' function than the one
developed in~\cite{HB09}, and then generalized to any $c \geq 0$.
We summarize our new results in Table~\ref{T-summary}.
\begin{table*}
   \centering
   \begin{tabular}{|c||c|c|c|c|}
      \hline
      $\F$ & $\C$ & $\K$ & $\V$ & $\R$ \\
      \hline
      $|V_\F^{(n)}(S)|$ & $O(n^{4+\eps})$ & $\Omega(n)$, $O(n^{2+\eps})$ &
         $\Omega(n^4)$, $O(n^{4+\eps})$ & $O(n^{4+\eps})$ \\
      \hline
      $|V_\F^{(f)}(S)|$ & $\Omega(n^4)$, $O(n^{4+\eps})$ & $O(n^{4+\eps})$ &
         $\Omega(n)$, $O(n^{4+\eps})$ & $\Omega(n^4)$, $O(n^{4+\eps})$ \\
      \hline
   \end{tabular} \\
   \begin{tabular}{|c||c|c|c|c|}
      \hline
      $\F$ & $\XS$ & $\XA$ & $\XP$ & $\PP$ \\
      \hline
      $|V_\F^{(n)}(S)|$ & $\Omega(n^4)$, $O(n^{4+\eps})$ &
         $\Omega(n^4)$, $O(n^{4+\eps})$ & $\Omega(n)$, $O(n^{4+\eps})$ &
         $O(n^{2+\eps})$ \\
      \hline
      $|V_\F^{(f)}(S)|$ & $\Omega(n^4)$, $O(n^{4+\eps})$ &
         $\Omega(n^4)$, $O(n^{4+\eps})$ & $\Omega(n^4)$, $O(n^{4+\eps})$ & \\
      \hline
   \end{tabular}
   \caption{Our results:  Worst-case combinatorial complexities of 2-site
            Voronoi diagrams of a set $S$ of $n$ points with respect to
            different distance functions}
   \label{T-summary}
\end{table*}

Throughout the paper we use the notation $V_\F^{(n)}(S)$ (resp.,
$V_\F^{(f)}(S)$) for denoting the nearest- (resp., furthest-) 2-site Voronoi
diagram, under the distance function $\F$, of a point set $S$.
The set $S$ is always assumed to contain $n$ points.


\section{Circumscribing Circle}

\label{S-circum-cir}

Let $\circ(p,q,r)$ denote the unique circle defined by three distinct
points $p$, $q$, and $r$ in the plane.  We now define the 2-site
circumscribing-circle distance function:
\begin{definition}
   Given two points $p,q$ in the plane, the ``circumcircle distance''
   $\C$ from a point $v$ in the plane to the unordered pair $(p,q)$ is
   defined as $\C(v,(p,q)) = \mathrm{Rad}(\circ(v,p,q))$.
\end{definition}
For a fixed pair of points $p$ and $q$, the curve
$\C(v,(p,q)) = \infty$ is the line $\overline{pq}$.
This implies that all the points on $\overline{pq}$ belong to
the region of $(p,q)$ in $V_{\C}^{(f)}(S)$. In this
section we assume that the points in $S$ are in general position,
i.e., there are no three collinear points, and no three pairs of
points define three distinct lines that intersect at one point.
The given sites are singular points, that is, for any two sites
$p,q$, the function $\C(v,(p,q))$ is not defined at $v=p$ or $v=q$.

\begin{theorem}
   Let $S$ be a set of $n$ points in the plane.
   The combinatorial complexity of $V_{\C}^{(f)}(S)$ is $\Omega(n^4)$.
\end{theorem}

\begin{proof}
   The $n$ points of $S$ define $\Theta(n^2)$ lines, which always have
   $\Theta(n^4)$ intersection points.  All these intersection points
   are features of $V_{\C}^{(f)}(S)$, and hence the lower bound.
\end{proof}

\begin{theorem}
   \label{TH-ub-circ}
   Let $S$ be a set of $n$ points in the plane.
   The combinatorial complexity of both $V_{\C}^{(n)}(S)$ and
   $V_{\C}^{(f)}(S)$ is $O(n^{4+\eps})$ (for any $\eps > 0$).
\end{theorem}

\begin{proof}
   Clearly, the combinatorial complexity of $V_{\C}^{(n)}(S)$ or
   $V_{\C}^{(f)}(S)$ is identical to that of the respective diagram of
   the 2-site distance function
   $\C^2(v,(p,q)) = \mathrm{Rad}^2(\circ(v,p,q))$.
   It is known that
   $\mbox{Rad}^2(\circ(v,p,q)) =
       ( (|vp| |vq| |pq|)/(4 |\triangle vpq|) )^2 =
       (((v_x-p_x)^2+(v_y-p_y)^2) ((v_x-q_x)^2+(v_y-q_y)^2)
          ((p_x-q_x)^2+(p_y-q_y)^2))/
          (4(v_x(p_y-q_y)-p_x(v_y-q_y)+q_x(v_y-p_y))^2).$
   The respective collection of $\Theta(n^2)$ Voronoi surfaces
   fulfills Assumptions~7.1 of~\cite[p.~188]{SA95}:
   \begin{enumerate}
   \item Each surface is an algebraic surface of maximum constant degree;
   \item Each surface is totally defined (this is stronger than needed); and
   \item Each triple of surfaces intersects in at most a constant number
         of points.
   \end{enumerate}
   Hence, we may apply Theorem~7.7 of~[ibid., p.~191] and
   obtain the claimed bound on the complexity of $V_{\C}^{(n|f)}(S)$.
\end{proof}


\section{Containing Circle}

\label{S-contain-cir}

Let $C(p,q,r)$ denote the minimum-radius circle
containing three points $p$, $q$, and $r$ in the plane.
(That it, $C(p,q,r)$ is the minimum circle containing the triangle
$\triangle pqr$.)
We now define the 2-site containing-circle distance function:
\begin{definition}
   Given two points $p,q$ in the plane, the ``containing-circle distance''
   $\K$ from a point $v$ in the plane to the unordered pair $(p,q)$ is
   defined as $\K(v,(p,q)) = \mathrm{Rad}(C(v,p,q))$.
\end{definition}

In our context we have that $p \neq q$.
Assume first that $v \ne p,q$.
Observe that if all angles of $\triangle pqr$ are acute (or $\triangle pqr$
is right-angled), then $C(p,q,r)$ is identical to $\circ(p,q,r)$.
Otherwise, if one of the angles of $\triangle pqr$ is obtuse, then
$C(p,q,r)$ is the circle whose diameter is the longest edge of
$\triangle pqr$, that is, the edge opposite to the obtuse angle.
If $v$ coincides with either $p$ or $q$, then $C(v,p,q)$ is the circle
whose diameter is the line segment $pq$.

\begin{theorem}
   Let $S$ be a set of $n$ points in the plane.
   The combinatorial complexity of $V_{\K}^{(n)}(S)$ is $\Omega(n)$.
\end{theorem}

\begin{proof}
   For simplicity assume that each point from $S$ has a unique
   closest neighbor in $S$.
   For each point $p \in S$, consider its closest neighbor $q$.
   Then, the points on the line segment $pq$ lying sufficiently close to
   $p$ belong to the region of $(p,q)$ in $V_{\K}^{(n)}(S)$, which is thus
   non-empty.  Since no region is thereby encountered more than twice,
   $V_{\K}^{(n)}(S)$ has at least $\lceil n/2 \rceil$ non-empty regions.
   The claim follows.
\end{proof}

\begin{theorem}
   Let $S$ be a set of $n$ points in the plane.
   The combinatorial complexity of $V_{\K}^{(n)}(S)$ is $O(n^{2+\eps})$
   (for any $\eps > 0$).
\end{theorem}

\begin{proof}
   Let a point $v$ belong to a non-empty region of $(p,q)$.
   No matter if the triangle $\triangle vpq$ is acute
   (Figure~\ref{fig:contain-NN}(a)),
   \begin{figure}
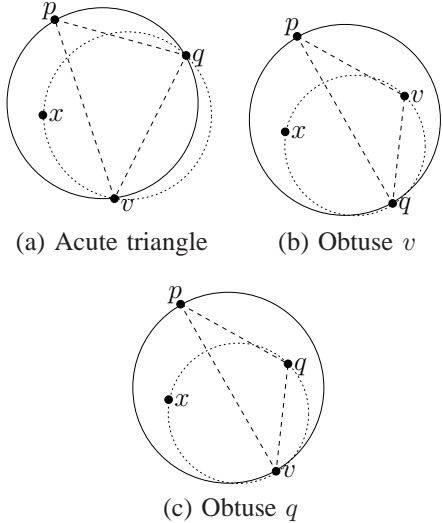

      \centering
      \begin{tabular}{cc}
         \scalebox{0.50}{\input{rad-con-a-BIG.tex}} &
            \scalebox{0.50}{\input{rad-con-b-BIG.tex}} \\
         (a) Acute triangle & (b) Obtuse $v$
      \end{tabular} \medskip \\
      \begin{tabular}{c}
         \scalebox{0.50}{\input{rad-con-c-BIG.tex}} \\
         (c) Obtuse $q$
      \end{tabular}
      \caption{If $p,q$ have a non-empty region in $V_{\K}^{(n)}(S)$,
               then $pq$ is an edge in DT$(S)$.}
      \label{fig:contain-NN}
   \end{figure}
   $\triangle vpq$ is obtuse with $v$ being the obtuse vertex
   (Figure~\ref{fig:contain-NN}(b)), or $\triangle vpq$ is obtuse with $p$
   or $q$ being the obtuse vertex (Figure~\ref{fig:contain-NN}(c)),
   the circle $C(v,p,q)$ cannot contain any other point $x \in S$.
   Otherwise, regardless of the location of $x$ in $C(v,p,q)$, we will
   always have $\K(v,(p,q)) > \K(v,(x,q))$, which is a contradiction.
   This follows from the fact (see~\cite[Lemma~4.14]{BKOS08}) that given
   a point set $K$ and its minimum enclosing circle $C$, where $C$ is
   defined by three points $a,b,c \in K$ (resp., two diametrical points
   $s,t \in K$), removing from $K$ one of $a,b,c$ (resp., one of $s,t$) will
   result in a point set with a smaller minimum enclosing circle.
   Thus, there is a circle containing $p,q$ that is empty of any other
   site from $S$.  This immediately implies that $pq$ is an edge
   of the Delaunay triangulation of $S$.  Consequently, there are $O(n)$
   pairs of sites in $S$ that have non-empty regions in $V_{\K}^{(n)}(S)$.
   Furthermore, it follows from the definition of $\K(v,(p,q))$ that
   the respective Voronoi surface of $(p,q)$ is made of a constant number
   of patches, each of which is a ``well-behaved'' function in the sense
   discussed in the proof of Theorem~\ref{TH-ub-circ}.  Again, by standard
   Davenport-Schinzel machinery, the combinatorial complexity of the lower
   envelope of these $O(n)$ surfaces is $O(n^{2+\eps})$ (for any $\eps > 0$),
   and the claim follows.
\end{proof}

\begin{theorem}
   Let $S$ be a set of $n$ points in the plane.
   The combinatorial complexity of $V_\K^{(f)}(S)$ is $O(n^{4+\eps})$
   (for any $\eps > 0$).
\end{theorem}

\begin{proof}
   As in the proof of Theorem~\ref{TH-ub-circ},
   we prove this claim by using the upper envelope of $\Theta(n^2)$
   ``well-behaved'' Voronoi surfaces.
\end{proof}


\section{View Angle}

\label{S-angle}

We now define the 2-site view-angle distance function:
\begin{definition}
   Given two points $p,q$ in the plane, the ``view-angle distance''
   $\V$ from a point $v$ in the plane to the unordered pair $(p,q)$ is
   defined as $\V(v,(p,q)) = \measuredangle{pvq}$.
\end{definition}
Similarly to the circumcircle-radius distance function, the view-angle
function is undefined at the $n$ given points.  For a fixed pair of
points $p$ and $q$, the curve $\V(v,(p,q)) = \pi$ is the open line
segment connecting the two points $p$ and~$q$, while the curve
$\V(v,(p,q)) = 0$ is the line $\overline{pq}$ excluding
the closed line segment $pq$.  The curve
$\V(v,(p,q)) = \pi/2$ is the circle whose diameter is the line
segment $pq$ (excluding, again, $p$ and $q$).

\begin{theorem}
   Let $S$ be a set of $n$ points in the plane.
   The combinatorial complexity of $V_\V^{(n)}(S)$ is $\Omega(n^4)$.
\end{theorem}

\begin{proof}
   Consider a set $S$ of $n$ points in the plane.  An example of the
   intersection of the complements of two segments defined by two pairs
   of points (with respect to the supporting lines) is shown in
   Figure~\ref{fig:NN-VD-V}(a).
   \begin{figure}
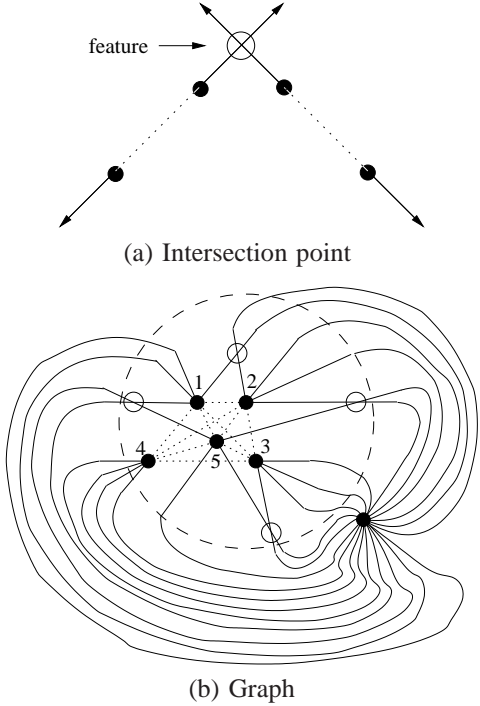

      \centering
      \begin{tabular}{c}
         \scalebox{0.75}{\input{view_nn.tex}} \\
         (a) Intersection point \medskip \\
         \scalebox{0.65}{\input{view_nn_2-new.tex}} \\
         (b) Graph
      \end{tabular}
      \caption{The graph of $V_\V^{(n)}(S)$.}
      \label{fig:NN-VD-V}
   \end{figure}
   These intersection points are features of $V_\V^{(n)}(S)$; we show
   that there are $\Omega(n^4)$ such points.  To this aim we create a
   geometric graph $G$ whose vertices are the
   given points, in which each segment's complement defines two edges.
   We add one additional point far away from the convex hull of $S$,
   and connect it (without adding intersections) to all the rays
   as shown in Figure~\ref{fig:NN-VD-V}(b).
   We can now use the crossing-number lemma for bounding from below the
   number of intersections of the original rays.
   The lemma tells us that every drawing of a graph with
   $n$ vertices and $m \geq 4n$ edges (without self or parallel edges)
   has $\Omega(m^3/n^2)$ crossing points~\cite{ACNS82,Le83}.  In our
   case $m = 2 \binom{n}{2} = n(n-1)$, so the number of intersection
   points in $G$ is $\Omega(n^6/n^2) = \Omega(n^4)$.
   All these intersection points are features of $V_\V^{(n)}(S)$,
   and hence the lower bound.
\end{proof}

\begin{theorem}
   Let $S$ be a set of $n$ points in the plane.
   The combinatorial complexity of both $V_\V^{(n)}(S)$ and $V_\V^{(f)}(S)$
   is $O(n^{4+\eps})$ (for any $\eps > 0$).
\end{theorem}

\begin{proof}
   For analyzing $V_\V^{(n)}$(S) and $V_\V^{(f)}(S)$ we consider the function
   $(-\cos \measuredangle{pvq})$ instead of that of $\measuredangle{pvq}$.
   This is permissible since the cosine function is strictly decreasing in
   the range $[0,\pi]$.
   By the cosine law, we have
   $-\cos \measuredangle{pvq} = (|pq|^2-|vp|^2-|vq|^2)/(2|vp||vq|)$.
   As we have already seen more than once in this paper,
   this means that the respective collection of $\Theta(n^2)$ Voronoi
   surfaces fulfills Assumptions~7.1 of~\cite[p.~188]{SA95}.
   Hence, we may apply Theorem~7.7 of [ibid., p.~191] and obtain the
   claimed bound on the complexity of $V_\V^{(n|f)}(S)$.
\end{proof}

\begin{theorem}
   \label{TH-lb-angle-fn}
   Let $S$ be a set of $n$ points in the plane.
   The combinatorial complexity of $V_\V^{(f)}(S)$ is $\Omega(n^4)$.
\end{theorem}

\begin{proof}
   Given a set $S$ of $n$ points in the plane, we count the intersections
   of pairs of line segments, where each segment is defined by points of
   $S$ (see Figure~\ref{fig:FN-VD-V}(a)).
   \begin{figure}
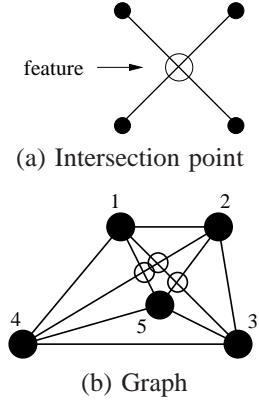

      \centering
      \begin{tabular}{c}
         \scalebox{0.75}{\input{view_fn.tex}} \\
         (a) Intersection point \medskip \\
         \scalebox{1.30}{\input{view_fn_2-new.tex}} \\
         (b) Graph
      \end{tabular}
      \caption{The graph of $V_\V^{(f)}(S)$.}
      \label{fig:FN-VD-V}
   \end{figure}
   We create a geometric graph whose vertices are the given points,
   and the edges are the line segments connecting every pair of points (see
   Figure~\ref{fig:FN-VD-V}(b)).
   The intersections of the segments defined by all pairs of points
   define features of $V_\V^{(f)}(S)$, because along these segments
   the view-angle function assumes its maximum possible value, $\pi$.
   We can now use the crossing-number lemma for counting these
   intersections.  The graph with $n$ vertices and $m \geq 4n$ edges
   (without self or parallel edges) has $\Omega(m^3/n^2)$ crossing
   points~\cite{ACNS82,Le83}.  In this case $m = \binom{n}{2} = n(n-1)/2$,
   hence $\Omega(n^4)$ is a lower bound on the complexity of $V_\V^{(f)}(S)$.
\end{proof}

Results by Asano et al.~\cite{AKTT06} immediately imply that the edges of
$V_\V^{(n|f)}(S)$ represent pieces of polynomial curves of degree at most
three.  However, the structure of the part of $V_\V^{(f)}(S)$ that lies
outside the convex hull $\mathcal{CH}(S)$ of $S$ is fairly simple: it is
given by the arrangement of lines supporting the edges of $\mathcal{CH}(S)$.
This arrangement can be computed by a standard incremental algorithm in
optimal $\Theta(k^2)$ time and space, where $k$ denotes the number of
vertices of $\mathcal{CH}(S)$.  Each cell of the arrangement should then be
labeled with a pair of sites from $S$, to the Voronoi region of which it
belongs; this extra task can be completed within the same complexity bounds.


\section{Radius of Inscribed Circle}

We now define the 2-site ``radius-of-inscribed-circle'' distance function:
\begin{definition}
   Given two points $p,q$ in the plane, the ``inscribed radius distance''
   $\R$ from a point $v$ in the plane to the unordered pair $(p,q)$,
   denoted by $\R(v,(p,q))$, is defined as the radius of the circle
   inscribed in the triangle $\triangle vpq$ (Figure~\ref{fig:inscr}).
\end{definition}

   \begin{figure}
      \centering
      \scalebox{0.70}{\input{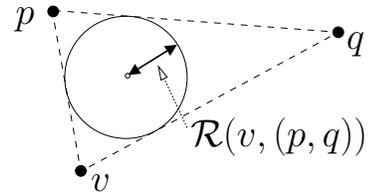}}
      \caption{$\R(v,(p,q))$ is the radius of the circle inscribed in
               $\triangle vpq$.}
      \label{fig:inscr}
   \end{figure}

\begin{theorem}
   \label{TH-lb-ins-rad-nn}
   Let $S$ be a set of $n$ points in the plane.
   The combinatorial complexity of $V_{\R}^{(n)}(S)$ is $\Omega(n^4)$.
\end{theorem}

\begin{proof}
   The intersection point of any two lines defined by the points from $S$
   is a distinct feature of the Voronoi diagram under discussion.
   Thus, $n$ points in $S$ define $\Theta(n^2)$ lines, which have
   $\Theta(n^4)$ intersection points.
\end{proof}

\begin{theorem}
   Let $S$ be a set of $n$ points in the plane.
   The combinatorial complexity of both $V_{\R}^{(n)}(S)$ and
   $V_{\R}^{(f)}(S)$ is $O(n^{4+\eps})$ (for any $\eps > 0$).
\end{theorem}

\begin{proof}
   Let $p,q$ be two points in $S$, and $v$ a point in the plane.
   It is a well-known fact that $\R(v,(p,q)) = 2\A(v,(p,q))/\PP(v,(p,q))$,
   where $\A(v,(p,q))$ and $\PP(v,(p,q))$ are the area and perimeter,
   respectively, of the triangle $\triangle vpq$.  Both the numerator and
   denominator of this fraction can be written as algebraic expressions
   using the coordinates of the points $v,p,q$.  Hence, as above, the
   standard Davenport-Schinzel machinery can be applied for obtaining the
   claim bounds.
\end{proof}

\begin{theorem}
   Let $S$ be a set of $n$ points in the plane.
   The combinatorial complexity of $V_\R^{(f)}(S)$ is $\Omega(n)$
   in the worst case.
\end{theorem}

\begin{proof}
   The complexity of $V_\R^{(f)}(S)$ can be as high as $\Omega(n)$.
   Let $S$ be a set of $n$ point in convex position with no three
   collinear points.
   Let $p$ and $q$ be two antipodal vertices of $\mathcal{CH}(S)$, the
   convex hull of $S$, and consider two parallel lines $\ell_p \ni p$ and
   $\ell_q \ni q$ tangent to $\mathcal{CH}(S)$ only at $p$ and $q$,
   respectively.
   Next, consider any point $v \in \ell_p$, and let it move along $\ell_p$
   in either direction.  In the limit, the distance from $v$ to any pair
   $(s,t)$ of sites in $S$ equals the width of the infinite strip bounded
   by two lines parallel to $\ell_p$ and passing through $s$ and $t$,
   respectively.  Consequently, the points of $\ell_p$ lying sufficiently
   far from $p$ belong to the Voronoi region of $(p,q)$.
   Since the number of pairs of antipodal vertices of $\mathcal{CH}(S)$
   is $\Theta(n)$, the bound follows.
\end{proof}

A similar reasoning leads to a conclusion that $V_{\R}^{(f)}(S)$ has at
most a linear number of unbounded regions.  To demonstrate this, consider
any point $u$ in the plane, and a line $\ell \ni u$.  Observe that the
points of $\ell$ lying sufficiently far from $u$ belong to the Voronoi
region of the pair(s) of points from $S$ that define the width of $S$ in
the direction orthogonal to $\ell$, and, thus, represent a pair (pairs)
of antipodal vertices of $\mathcal{CH}(S)$.  Since the union of all such
lines gives the whole plane, and the number of antipodal vertices of
$\mathcal{CH}(S)$ is at most linear, the claim follows.


\section{Distances Based on the Center of the Circumscribing Circle}

Let $v,p,q$ be three points in the plane. Consider the circle $\circ(v,p,q)$
passing through $v,p,q$ with center $o_{vpq}$.
We now define three more distance functions based on the above notation:
\begin{definition}
   Given two points $p,q$ in the plane, the three distances, denoted by
   $\XS(v,(p,q))$, $\XA(v,(p,q))$, and $\XP(v,(p,q))$, respectively, are
   the distance from $o_{vpq}$ to the line segment $pq$, the area of the
   triangle $\triangle o_{vpq}pq$, and the perimeter of
   $\triangle o_{vpq}pq$, respectively (Figure~\ref{fig:circ-based}).
\end{definition}

   \begin{figure}
      \centering
      \scalebox{0.55}{\input{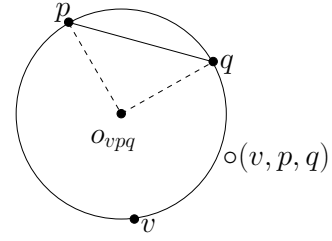}}
      \caption{The circle $\circ(v,p,q)$ is defined by the points $v,p,q$,
               and has the center at $o_{vpq}$.  $\XS(v,(p,q))$ is the
               distance from $o_{vpq}$ to the segment $pq$ (or,
               equivalently, the height of $\triangle o_{vpq}pq$
               perpendicular to $pq$), and $\XA(v,(p,q))$, and
               $\XP(v,(p,q))$ are the area and the perimeter of $\triangle
               o_{vpq}pq$, respectively.}
      \label{fig:circ-based}
   \end{figure}

The upper bound of $O(n^{4+\eps})$ (for any $\eps > 0$) on the complexity of
the nearest- and furthest-neighbor Voronoi diagrams under each of these
distance functions can be, again, derived by means of Davenport-Schinzel
machinery.  Below we provide some lower bounds.  First, we address the
nearest-neighbor case.

\begin{theorem}
   Let $S$ be a set of $n$ points in the plane.
   The combinatorial complexity of $V_\XS^{(n)}(S)$ and $V_\XA^{(n)}(S)$
   is $\Omega(n^4)$ in the worst case.
\end{theorem}

\begin{proof}
   The key observation is the following.  Consider a pair $(p,q)$ of sites,
   and let $\circ(p,q)$ denote the circle with the diameter $pq$.  Then, for
   any point $v \in \circ(p,q) \setminus \{p,q\}$, we have
   $\XS(v,(p,q)) = \XA(v,(p,q)) = 0$.

   Consider two parallel lines $l_1$ and $l_2$, and let $d$ denote the
   distance between them.  For a given $n\ge 2$, let us construct a set $S$
   of $n$ points as a union of two sets $S_1\subset l_1$ and $S_2\subset l_2$
   consisting of $\lceil n/2\rceil$ and $\lfloor n/2 \rfloor$ points,
   respectively, in the following way.  The sets $S_1$ and $S_2$ are
   constructed iteratively; at each odd step, a new point is added to $S_1$,
   and at each even one---to $S_2$.  For any $i$: $2\le i\le n$, let $S_1^i$
   and $S_2^i$ denote the two sets constructed so far, and let
   $M^i=\{\circ(p,q) | p\in S_1^i, q\in S_2^i\}$ denote the set of circles
   defined by pairs of points from different sets.
   We want each circle from $M^n$ to pass through precisely two points from $S$ (those defining it), each two circles from $M^n$ to intersect, and no three of them to
   pass through the same point not contained in $S$. Then $\Theta(n^2)$ circles composing $M^n$
   will give rise to $\Theta(n^4)$ distinct intersection points, each
   belonging to a separate feature of either Voronoi diagram under
   consideration, and the claim will follow.

   To ensure the first property, we select the points so that the distance
   between each two points contained in the same set $S_i$ is much smaller
   than $d$, where $i=1,2$. To guarantee the second property, at each step
   $j$: $3\le j\le n$, when adding a new point $s$ to the respective set,
   we make sure that for any point $t$ from the other set, the circle
   $\circ(s,t)$ passes neither through any point from $S_1^{j-1}\cup S_2^{j-1}\setminus \{t\}$ nor through any intersection point of the
   circles from~$M^{j-1}$.
   This completes the proof.
\end{proof}

\begin{theorem}
   Let $S$ be a set of $n$ points in the plane.
   The combinatorial complexity of $V_\XP^{(n)}(S)$ is $\Omega(n)$ in the
   worst case.
\end{theorem}

\begin{proof}
   A linear lower bound in the worst case for $V_\XP^{(n)}(S)$ can be
   obtained in the following way.  Choose the set $S$ of points to lie on
   some line $\ell$, so that the distance between any two consecutive
   points is~1. Then, the minimum possible value for the distance function
   $\XP$ is obviously~2, and can be achieved only for a pair $(p,q)$ of
   consecutive points.  For each such pair $(p,q)$, consider the circle
   $\circ(p,q)$ with the diameter $pq$.  Evidently,
   for any point $v \in \circ(p,q) \setminus \{p,q\}$, we have
   $\XP(v,(p,q)) = 2$, and for any other pair $(s,t)$ of sites,
   $\XP(v,(s,t)) > 2$.  We conclude that each pair of consecutive points
   along $\ell$ has a non-empty region in $V_\XP^{(n)}(S)$.  Since there
   are $n-1$ pairs of consecutive points, the bound follows.
\end{proof}

Second, we address the furthest-neighbor Voronoi diagrams.

\begin{theorem}
   Let $S$ be a set of $n$ points in the plane.
   The combinatorial complexity of all of $V_\XS^{(f)}(S)$,
   $V_\XA^{(f)}(S)$, and~$V_\XP^{(f)}(S)$ is $\Omega(n^4)$.
\end{theorem}

In each case, the proof is identical to that of
Theorem~\ref{TH-lb-ins-rad-nn}.


\section{Parameterized Perimeter}

Finally, we define the 2-site parameterized perimeter distance function:
\begin{definition}
   Given two points $p,q$ in the plane and a real constant $c \geq -1$,
   the ``parameterized perimeter distance'' $\PP_c$ from a point $v$ in
   the plane to the unordered pair $(p,q)$ is
   defined as $\PP_c(v,(p,q)) = |vp|+|vq|+c\cdot |pq|$.
\end{definition}

We require that $c$ be greater than or equal to $-1$ since allowing
$c < -1$ would result in negative distances.  Letting $c=-1$ results in a
distance function that equals~0 for all the points on the line segment~$pq$.
 If $c=0$, we deal with the ``sum of distances'' distance function
introduced in~\cite{BDD02} and recently revisited in~\cite{VB10}.
For $c=1$, the above definition yields the ``perimeter'' distance function
$\PP(v,(p,q)) = \mathrm{Per}(\triangle vpq)$.

In~\cite{HB09} it was proven that the combinatorial complexity of the
nearest-neighbor 2-site perimeter Voronoi diagram of a set of $n$ points is slightly
superquadratic in $n$.  In a nutshell, the proof was based
on the observation that any pair of sites that has a non-empty region in
the perimeter diagram also has a non-empty region in the sum-of-distances
diagram.  This immediately implies that the number of such pairs is
linear in $n$.  (However, unlike in the sum-of-distances
diagram, a region in the perimeter diagram is not necessarily continuous.
We were able to construct examples in which the number of connected
components of a \emph{single} region is comparable to the number of points!)
Again, one can apply the standard Davenport-Schinzel machinery and
conclude the claimed upper bound on the complexity of the diagram.
It remains unclear whether the worst-case complexity of the diagram is
linear, quadratic, or in between.
The proof in~\cite{HB09} of the main observation was extremely complex.
We provide here an alternative and much simpler proof of the same bound,
which generalizes to the case of ``parameterized perimeter'' distance
function for any $c \geq 0$.

\begin{theorem}
   \label{TH-ub-per}
   Let $S$ be a set of $n$ points in the plane.
   The combinatorial complexity of $V_\PP^{(n)}(S)$ is $O(n^{2+\eps})$
   (for any $\eps > 0$).
\end{theorem}

\begin{proof}
   Refer to Figure~\ref{fig:NN-VD-P}.
   \begin{figure}
      \centering
      \begin{tabular}{c}
         \input{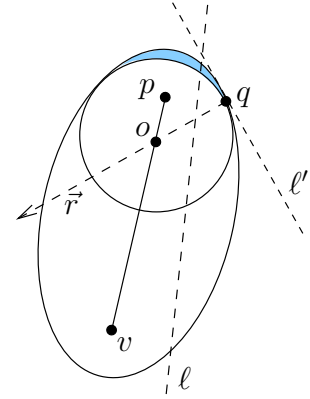}
      \end{tabular}
      \caption{An empty circle containing sites in $\PP$.}
      \label{fig:NN-VD-P}
   \end{figure}
   Let $p,q \in S$ be two sites which have a non-empty region in
   $V_\PP^{(n)}(S)$, and let $v$ be a point in this region, noncollinear with $p$ and $q$.
   In addition, let $\ell$ be the perpendicular bisector of the
   line segment $pq$.
   Assume, without loss of generality, that $|vp| \leq |vq|$.

   Consider the ellipse $O_{vpq}$ passing through $q$ with $v$
   and $p$ as foci.  By definition, for any point $s$ inside this ellipse
   we have $|vs|+|ps| < |vq|+|pq|$.  Therefore,
   \begin{eqnarray}
      \label{EQ-per}
      \PP(v,(p,s)) \! \! \! & = \! \! \! & |vs|+|ps|+|vp| \\
      \nonumber
                            & < \! \! \! & |vq|+|pq|+|vp| \, = \, \PP(v,(p,q)).
   \end{eqnarray}
   This means that $s$ cannot be a site in $S$, for otherwise $v$ would
   belong to the region of $(p,s)$ instead of to the region of $(p,q)$.
   It follows that the ellipse $O_{vpq}$ is empty of any sites other
   than $p$ and $q$.

   Now consider the line $\ell'$ that is tangent to $O_{vpq}$ at $q$, and
   the ray $\vec{r}$ perpendicular to $\ell'$ at $q$ and
   passing through $O_{vpq}$.  It is a known property of ellipses that
   this ray bisects the angle $\measuredangle{vqp}$, and, thus, it
   intersects the line segment $vp$, say, at point $o$.
   The circle $C$ centered at $o$ and passing through~$q$ is tangent to
   $O_{vpq}$ at $q$ (as well as at another point), and is entirely
   contained in $O_{vpq}$.  Since $v$ is closer to $p$ than to $q$ (by our
   assumption), it follows that the circle $C$ also contains~$p$.  (If $p$
   were on the extension of $vp$ in the shaded area, a contradiction would
   easily be obtained by using the triangle inequality:  $|op|>|oq|$, hence
   $|vp|=|ov|+|op|>|ov|+|oq|>|vq|$, contradicting the assumption that
   $|vp| \leq |vq|$.)  Since $O_{vpq}$ is empty of sites (except $p$ and $q$),
   so~is the circle~$C$.
   Therefore, $pq$ is an edge of the Delaunay triangulation of~$S$.
   The number of such edges is linear in $n$, the cardinality of~$S$.

   Hence, there are $\Theta(n)$ respective surfaces of these pairs of sites.
   One can now apply the standard Davenport-Schinzel machinery (as in the
   proof of Theorem~\ref{TH-ub-circ}).  The claim follows.
\end{proof}

Finally, we state the following theorem.

\begin{theorem}
   \label{TH-param-per}
   Let $S$ be a set of $n$ points in the plane.
   \begin{itemize}
   \item[(a)]
      The combinatorial complexity of $V_{\PP_{-1}}^{(n)}(S)$ is
      $\Omega(n^4)$ and $O(n^{4+\eps})$ (for any $\eps > 0$).
   \item[(b)]
      If there is a unique closest pair $p,q \in S$, then when
      $c \to \infty$, the combinatorial complexity of $V_{\PP_c}^{(n)}(S)$
      is asymptotically~1.
   \item[(c)]
      For $c \geq 0$, the combinatorial complexity of $V_{\PP_c}^{(n)}(S)$
      is $O(n^{2+\eps})$ (for any $\eps > 0$).
   \end{itemize}
\end{theorem}

\begin{proof}
   \begin{itemize}

   \item[(a)]
      To see the lower bound on the complexity of $V_{\PP_{-1}}^{(n)}(S)$,
      note that every point
      on the segment $pq$ has $\PP_{-1}$-distance zero to the pair $(p,q)$,
      and therefore, the intersection of any pair of segments $p_1 q_1$ and
      $p_2 q_2$ defined by sites $p_1, q_1, p_2, q_2 \in S$ is a feature of
      $V_{\PP_{-1}}^{(n)}(S)$.
      As is demonstrated in the proof of Theorem~\ref{TH-lb-angle-fn},
      the number of these features is $\Omega(n^4)$.
      The upper bound is obtained by using the usual Davenport-Schinzel
      machinery, as in the proof of Theorem~\ref{TH-ub-circ}.

   \item[(b)]
      It is easy to verify that as $c \to \infty$, the term $c \cdot |pq|$
      dominates the distance $\PP_c(v,(p,q))$, and, hence, every point $v$
      in the plane is closer to the unique closest pair of sites
      $p,q \in S$ than to any other pair in $S$.  Hence, the asymptotic
      diagram contains zero vertices, zero edges, and one face (the entire
      plane).

   \item[(c)]
      The proof is a generalized version of the proof of the special case
      $c=1$.  Refer to Figure~\ref{fig:oval}.
   \begin{figure}
      \centering
      \scalebox{0.60}{\input{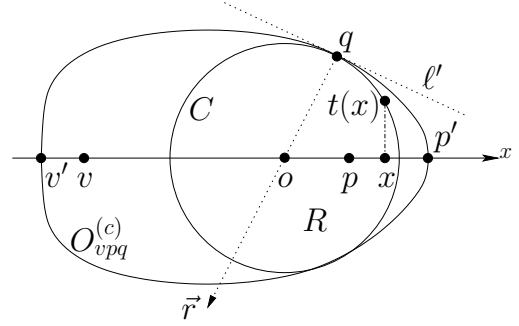}}
      \caption{The Cartesian oval $O_{vpq}^{(c)}$ is the locus of points $q'$,
      for which $|vq'|+c\cdot |pq'|=|vq|+c\cdot |pq|$. The ray $\vec{r}$
      passes through $q$ and is perpendicular to $O_{vpq}^{(c)}$,
      and intersects the axis of symmetry of $O_{vpq}^{(c)}$ at the point $o$.
      The circle $C$ is centered at $o$, and is tangent to $O_{vpq}^{(c)}$ at
      $q$.
      For any point $x$ on the axis of abscissas residing inside $C$, $t(x)$
      denotes the point of~$C$ lying above $x$.}
      \label{fig:oval}
   \end{figure}
      As in the proof of Theorem~\ref{TH-ub-per}, we assume that there
      is a point $v$ in the region of $(p,q)$, such that $|vp|\le |vq|$, and
      $v$ is noncollinear with $p$ and $q$.
      Our goal is to show that for any $c \geq 0$ there exists a circle
      having $q$ on its boundary and containing $p$, which is empty of any
      other site $s$, implying that $p,q$ are Delaunay neighbors.

      As in the proof of Theorem~\ref{TH-ub-per}, let $O^{(c)}_{vpq}$ be
      the locus of points $q'$ for which $\PP_c(v,(p,q'))=\PP_c(v,(p,q))$.
      Thus, $O^{(c)}_{vpq}$ is the \emph{Cartesian oval} $(v,p,c,k)$
      consisting of all points $q'$ that satisfy $|vq'|+c|pq'|=k$, where
      $k=|vq|+c|pq|$ is constant.  (Unless $c=1$, this oval has exactly
      one axis of symmetry: the line joining the two foci $v,p$.)
      Then, if there were a site $s$ within $O^{(c)}_{vpq}$, it would lead
      to a smaller value of $\PP_c$, so $O^{(c)}_{vpq}$ must be empty of
      sites other than $p$.

      As before, let $\vec{r}$ be the ray emanating from $q$ perpendicular
      to and pointing into $O^{(c)}_{vpq}$, and let $o$ be the point where
      $\vec{r}$ crosses the line $pv$.

      Let us further suppose that $c\ne 1$. Without loss of generality,
      assume that $O_{vpq}^{c}$ is symmetric with respect to the axis of
      abscissas (see Figure~\ref{fig:oval}); consequently, the points $p$,
      $v$, and $o$ belong to the latter.  Let $x_p$, $x_v$, $x_o$, and $x_q$
      denote the corresponding coordinate of $p$, $v$, $o$, and $q$,
      respectively.

      Consider a circle $C$ centered at $o$ of the radius $R=|oq|$.  By
      construction, $C$ is tangent to $O_{vpq}^{(c)}$ at $q$.

      For any $x\in \mathbb{R}$, such that the point $(x,0)$ lies inside
      $C$, let $t(x)$ denote the point of $C$ lying above $(x,0)$.  For any
      such $x$, let
      \begin{align*}
      f_v(x)&=d(v,t(x))\\
            &=\sqrt{R^2-(x-x_o)^2+(x-x_v)^2}\\
            &=\sqrt{2(x_o-x_v)\cdot x +x_o^2+x_v^2+R^2}.
      \end{align*}
      Since $f_v(x)$ represents a square root of a linear function,
      it is concave on its domain.  The same will hold for a function
      $f_p(x)=d(p,t(x))$.  Consequently, their weighted combination
      $f(x)=f_v(x)+c\cdot f_p(x)$ is also concave on the same domain, and,
      thus, has a single local maximum.

      Recall that the circle $C$ is tangent to $O_{vpq}^{c}$ at $q$ by
      construction.
      It is easy to see that $C$ is tangent to $O_{vpq}^{c}$ \emph{from
      the inside}: otherwise, $x_q$ would be a local minimum of $f(x)$
      achieved at an inner point of the domain, contradicting the concavity
      of $f(x)$.  It follows that $f(x)$ has a local maximum at $x_q$.
      Together with the previous observation, this implies that $f(x)$ has
      a global maximum at $x_q$. This means that $q$ is the only common
      point of $O_{vpq}^{(c)}$ and the upper half of $C$.  By symmetry, we
      conclude that $C$ lies inside $O_{vpq}^{(c)}$ and touches it at $q$
      and the point symmetric to $q$.  Thus, $C$ must be empty of sites
      other than $p$.

      It remains to demonstrate that $p$ lies inside $C$.  To this end, it
      is sufficient to show that the point $o$ lies between $v$ and $p$;
      then, as in the case of $c=1$, the needed property can be easily
      derived using the triangle inequality.

      Let us argue as follows.  The above reasoning can be carried out for
      any point $q'\in O_{vpq}^{c}$ noncollinear with $v$ and $p$, providing
      us with a maximum empty circle inscribed in $O_{vpq}^{c}$, and tangent
      to it at precisely two points---namely, at $q'$ and its symmetric
      point.  It follows that the medial axis of~$O_{vpq}^{c}$ is a segment
      of the line $\overline{vp}$ through $v$ and $p$.
      Let $v'$ and $p'$ be the intersection points of $\overline{vp}$ and
      $O_{vpq}^{(c)}$ being closer to $v$ and~$p$, respectively (see
      Figure~\ref{fig:oval}).  Consider the circle~$C_v$ with radius $|vv'|$
      centered at~$v$.  Obviously, $v'$ is a common point of $C_v$ and
      $O_{vpq}^{(c)}$, but any other point $z$ of $C_v$ lies strictly inside
      $O_{vpq}^{c}$, since for any such point $z$, we have $|zv|=|v'v|$ and
       $|zp|<|v'p|$.
      This implies that the radius of curvature of $O_{vpq}^{(c)}$ at $v'$
      is greater than $|vv'|$.  A similar statement holds for~$p'$.
      Consequently, the two endpoints of the medial axis must lie between
      $v$ and $p$, and the same must hold for the point $o$.

      We conclude that $C$ is a circle containing both $p$ and $q$ and
      otherwise empty of sites, so $p$ and $q$ are Delaunay neighbors.
      Hence, there are $\Theta(n)$ pairs of sites that generate regions in
      the Voronoi diagram, and the claim follows from the standard
      Davenport-Schinzel machinery.
   \end{itemize}
\end{proof}


\section{Conclusion}

\label{S-conclusion}

In this paper, we have investigated 2-site Voronoi diagrams of point sets
with respect to a few geometric distance functions.  The Voronoi structures
obtained in this way cannot be explained in terms of the previously known
kinds of Voronoi diagrams (which is the case for the 2-site distance
functions thoroughly analyzed in~\cite{BDD02}), what makes them particularly
interesting.  On the other hand, our results can be exploited to advance
research on Voronoi diagram for segments.  Potential directions for future
work include consideration of other distance functions, and generalizations
to higher dimensions and to $k$-site Voronoi diagrams.


\section*{Acknowledgments}

Work on this paper by the first author was performed while he was
affiliated with Tufts University in Medford, MA.
Work by the last author was partially supported by
Russian Foundation for Basic Research (grant~10-07-00156-a).


\end{document}